\documentclass[letterpaper,12pt,reqno]{article}

\usepackage[margin=1.1in]{geometry}
\usepackage[T1]{fontenc}
\usepackage[english]{babel}

\usepackage{amsmath,amssymb,amsfonts,amsthm}
\usepackage{mathtools}
\usepackage{stmaryrd} 

\usepackage{microtype}
\allowdisplaybreaks

\usepackage{graphicx}
\usepackage{tikz}
\usetikzlibrary{decorations.pathreplacing,arrows.meta}

\usepackage{algorithm}
\usepackage{algpseudocode}
\usepackage{float} 

\usepackage{xcolor}
\definecolor{Myblue}{rgb}{0,0,0.6}
\definecolor{cellviolet}{RGB}{135,55,195}
\definecolor{dotblue}{RGB}{40,50,185}
\usepackage[colorlinks,citecolor=Myblue,linkcolor=Myblue,urlcolor=Myblue,pdfpagemode=None]{hyperref}

\theoremstyle{plain}
\newtheorem{theorem}{Theorem}[section]

\newtheorem{lemma}[theorem]{Lemma}
\newtheorem{corollary}[theorem]{Corollary}

\theoremstyle{definition}
\newtheorem{example}[theorem]{Example}
\newtheorem{definition}[theorem]{Definition}
\newtheorem{remark}[theorem]{Remark}

\newcommand{\bb}[1]{\mathbb{#1}}

\title{Solomonoff Induction and Singular Integrals}
\author{
Will Troiani\thanks{Resolution. \texttt{will@resolution.org}}
\and
Daniel Murfet\thanks{Resolution. \texttt{murfet@resolution.org}}
}
\date{\today}

\begin{document}

\maketitle

\begin{abstract}
The Solomonoff distribution $M$ assigns an a priori probability to a finite
binary string $z$ by summing over all programs whose output begins with $z$, a
program of length $\ell$ weighted by $2^{-\ell}$: the likely strings are those
with many short explanations. This is apparently a purely discrete notion of
complexity. However, Riemann sums are also countable. We show that the sum
defining $M$ contains, suitably reorganised, Riemann sums approximating the
Bayesian evidence of any computable statistical model --- singular integrals
whose asymptotics are governed by Singular Learning Theory and, through it, by
invariants of algebraic geometry. Concretely, for every computable Bayesian
model we construct a single monotone Turing machine whose induced semimeasure
agrees with the evidence $Z_n$ up to a uniform multiplicative constant. For a
computable Bayesian model satisfying in addition the hypotheses of Watanabe's
free energy asymptotics, it follows that data $X^n = X_1\cdots X_n$ drawn
i.i.d.\ from the true distribution satisfies
\[
-\log M(X^n)\ \le\ n L_n(w_0) + \lambda\log n - (m-1)\log\log n + O_{\bb P}(1),
\]
where $L_n$ is the empirical loss, $w_0$ an optimal parameter, $\lambda$ the
learning coefficient, and $m$ its multiplicity. The learning coefficient, a
geometric measure of model simplicity, thereby appears within the Solomonoff
distribution as the coefficient of $\log n$ in an upper bound on code length.
\end{abstract}

\section{Introduction}
\label{sec:intro}

An approach to universal Artificial General Intelligence (AGI) has been defined
in \cite{Hutter2005,Hutter:24uaibook2} using a formalisation of an agent
performing optimally across all computable environments. A key ingredient towards this definition is the \emph{Solomonoff prior}, which is the mixture,
defined with respect to a fixed universal monotone Turing machine $\mathcal{U}$:
\[
\xi_{\mathcal{U}} := \sum_{\nu} 2^{-\mathcal{K}(\nu)} \cdot \nu
\]
where the sum is over all lower-semicomputable semimeasures $\nu$, and $\mathcal{K}$ is the
Kolmogorov complexity with respect to $\mathcal{U}$.\footnote{Our monotone Turing machines have
tape alphabet $\bb{B}=\{0,1\}$ and no blank symbols, so the contents of the input tape may be
read as random bits, to be used for random sampling (Definition~\ref{def:monotone-tm}).
In the context of reinforcement learning, this sum is instead over
lower-semicomputable \emph{chronological} semimeasures, but we will not consider these
here.} Once a method of evaluating the performance of agents has been
constructed, the universally intelligent agent AIXI (sometimes spelt AI$\xi$) is then defined
as a maximiser of this valuation; see
\cite[Chapter~7]{Hutter:24uaibook2} for details. A key property of AIXI is that it uses
\emph{simple} reasoning, as motivated by philosophy of science concepts such as Occam's
razor\footnote{Among the hypotheses consistent with the observations, prefer the simplest
\cite[Section~16.1]{Hutter:24uaibook2}.}
and Epicurus' principle\footnote{Discard no hypothesis consistent with the observations; in
the prior, every such hypothesis retains positive weight
\cite[Section~16.1]{Hutter:24uaibook2}.}, and formalised by the minimisation
of $\mathcal{K}$.

It is therefore common to swap out the Solomonoff prior for the Solomonoff \emph{distribution} $M$,
which agrees with $\xi_{\mathcal{U}}$ up to a multiplicative constant in the case
of sequence prediction, and is defined as follows for a finite binary string $z$:
\begin{equation}\label{eq:M_sum}
M(z) = \sum_{p\,:\,\mathcal{U}(p)=z\ast} 2^{-\ell(p)},
\end{equation}
summing over \emph{minimal programs} $p$; here $\mathcal{U}(p)=z\ast$ means that in the run of
$\mathcal U$ on any input tape extending $p$, the input head has consumed exactly $p$ at the
step at which the last bit of $z$ is written (Definition~\ref{def:monotone-tm}). That
\eqref{eq:M_sum} is a reasonable notion of a priori probability is exemplified by the
correspondence between machines and semimeasures: every monotone Turing machine $T$ induces a
lower-semicomputable semimeasure $P_T$ (Definition~\ref{def:program-semimeasure}), where
$P_T(z)$ is the probability that the output of $T$ begins with $z$ when its input tape
is filled with uniformly random bits, and conversely every lower-semicomputable semimeasure arises this
way up to a bounded multiplicative factor \cite{Hutter2005,Hutter:24uaibook2}. Under this
correspondence, $M$ dominates each $P_T$ with constant $2^{-\ell(\langle T\rangle)}$
(Theorem~\ref{thm:dominance}).

The minimal programs $p$ in \eqref{eq:M_sum} are binary strings which may be of
the form $\langle T \rangle p'$ where
$\langle T \rangle$ is some encoding of a Turing machine $T$ and $p'$ is the rest of the
string, which may be passed as input to $T$ or as extra information. Amongst such strings are
two conceptually distinct categories: \emph{deterministic programs}, specified by the
portion of the string $\langle T \rangle$ and some input to $T$ contained in $p'$; and \emph{samplers},
which read their code from $\langle T \rangle$ and then continue reading $p'$ as
random bits used to sample from a predictive distribution.

A preference for simplicity has also been considered in a different
context, that of \emph{Bayesian learning}. Fix a model
class $\{p(\cdot \mid w) \mid w \in W\}$,
a family of probability distributions indexed by a parameter $w$ ranging over a space $W$,
together with a prior distribution $\phi$ on $W$ with density $\varphi$. Given data $D_n = (z_1, \ldots, z_n)$
consisting of $n$ i.i.d.\ samples, Bayes' rule updates the prior to the \emph{posterior}
distribution
\[
p(w \mid D_n) \propto \varphi(w)\prod_{i=1}^{n} p(z_i \mid w),
\]
and by Bayesian \emph{learning} we mean the evolution of this posterior as $n$ increases.
Amongst the samplers appearing in \eqref{eq:M_sum} are those implementing Bayesian
learning for a fixed computable model class: draw $w \sim \phi$, then sample from
$p(\cdot \mid w)$.

When the model class is \emph{singular}, Singular Learning Theory (SLT) describes the
large-$n$ behaviour of the \emph{Bayesian evidence}
\[
Z_n := \int_W \varphi(w)\prod_{i=1}^{n} p(z_i \mid w)\,dw,
\]
the normalising constant of the posterior above. Watanabe's theorems
\cite{watanabeAlgebraicGeometryStatistical2009,Watanabe2018} give the asymptotic expansion of
the associated \emph{free energy} (see Appendix~\ref{sec:setup:slt}):
\[
F_n := -\log Z_n = n L_n(w_0) + \lambda \log n - (m-1)\log\log n + O_{\bb P}(1),
\]
where $L_n(w) := -\tfrac{1}{n}\sum_{i=1}^n \log p(z_i \mid w)$ is the \emph{empirical loss},
$w_0$ is an optimal parameter, $\lambda > 0$ is the \emph{learning coefficient}, and $m$
is its multiplicity. A localisation of $\lambda$ near
a chosen parameter, the Local Learning Coefficient (LLC), has been adopted as a measure of
model simplicity
\cite{lau2024locallearningcoefficientsingularityaware,hoogland2024developmental}: in the
regular case the learning coefficient equals half the number of parameters, and so reads as an
effective parameter count. This leaves us with two notions of simplicity: minimum description
length on the one hand, and singular model simplicity as measured by the learning coefficient
on the other.

The goal of the present paper is to see how these two notions of simplicity
relate to one another. At first sight they belong to different worlds: the sum
\eqref{eq:M_sum} is indexed by a countable set of discrete objects, while the
evidence $Z_n$ is an integral over a continuous parameter space, whose
asymptotics are controlled by the geometry of the zero locus of the
Kullback--Leibler divergence and ultimately by resolution of singularities.
The bridge is the observation that Riemann sums are also countable. We show
that the sum defining $M$ can be reorganised so that within it appear Riemann
sums for the evidence integrals of computable Bayesian models.

It is no great surprise that Bayesian learning can be approximated by
computable procedures, and so ought to contribute terms to the sum
\eqref{eq:M_sum}. There is a
challenge, though: to arrive at a meaningful relationship between $M(z)$ on the one hand, with
$z = z_1 \cdots z_n$ and $n$ free to grow, and the learning coefficient on the other, the
contribution must come from a \emph{single fixed} machine serving every $n$. This is possible,
and constitutes the main technical work of the paper.

The construction has two stages. First, we define a computable procedure
$\mathsf{BayesSampler}(\mathfrak m,\omega)$ (Definition~\ref{def:bayes-sampler}), where
$\mathfrak m$ is a finite effective description of the statistical model, comprising an
integer $b > 0$, a model class $\{p(\cdot\mid w) \mid w\in W\}$ of probability distributions on
the set $\bb{B}^b$ of bit-strings of length $b$, and a prior $\phi$ on $W$, and
$\omega\in\bb{B}^{\bb{N}}$ is an infinite string of uniformly random bits. Second, fixing
$\mathfrak m$ and
hard-wiring it into the procedure yields an ordinary monotone Turing machine $B_{\mathfrak m}$
with $B_{\mathfrak m}(\omega)=\mathsf{BayesSampler}(\mathfrak m,\omega)$, whose only input is
$\omega$. This one machine is fixed by $\mathfrak m$ alone, receives neither the sample size
nor the data, runs forever, and writes a single infinite string; its induced semimeasure
reproduces the Bayesian evidence $Z_n$ of the model on every prefix, up to a uniform constant
(Theorem~\ref{thm:sampler-correct}). The relationship is that of a Riemann sum to its
integral: for every string $z = z_1 \cdots z_n$ consisting of $n$ strings
$z_i \in \bb{B}^b$,
\begin{alignat*}{2}
P_{B_{\mathfrak m}}(z)
&=
\sum_{\alpha}
&&\prod_{t=1}^{n} p\bigl(z_t\mid w_t(\alpha)\bigr)\,\phi\bigl(D_{j(n),\alpha}\bigr),
\\
Z_n
&=
\int_W
&&\prod_{t=1}^{n} p\bigl(z_t\mid w\bigr)\,d\phi(w),
\end{alignat*}
where the cells $D_{j(n),\alpha}$ partition the parameter space $W$, up to a null set of
boundaries, at a resolution growing with $n$, and $w_t(\alpha)$ is the centre of the cube
of the current cell at time step $t$ (Lemma~\ref{lem:sampler-identity}). The estimate relating the two is 
\[
e^{-\pi^2/6}\,Z_n\ \le\ P_{B_{\mathfrak m}}(z)\ \le\ e^{\pi^2/6}\,Z_n,
\]
so the induced semimeasure of $B_{\mathfrak m}$ has exactly the free energy asymptotics of
the model (Corollary~\ref{cor:sampler-free-energy}). Since $B_{\mathfrak m}$ is
one monotone Turing machine
among those contributing to \eqref{eq:M_sum}, dominance yields the main 
result (Theorem~\ref{thm:universal-aware}):
\[
-\log M(z)\ \le\ -\log Z_n + O(1),
\]
and when the strings $X_1,\ldots,X_n$ are drawn i.i.d.\ from a true distribution satisfying
the hypotheses of Watanabe's theorem (Appendix~\ref{sec:setup:slt}), writing
$X^n := X_1\cdots X_n$,
\[
-\log M(X^n)\ \le\ n L_n(w_0) + \lambda\log n - (m-1)\log\log n + O_{\bb P}(1).
\]
The two notions of simplicity are thereby related within a single bound: the additive
constant, which depends on the model description but on neither $n$ nor the data, pays for
the description of the model, while the coefficient of $\log n$ is its learning coefficient.

Background on the Solomonoff distribution, following \cite{Hutter:24uaibook2}, is recalled in
Appendix~\ref{sec:setup:universal}, and the necessary singular learning theory in
Appendix~\ref{sec:setup:slt}. Related work is discussed in Remark~\ref{rem:related-work}.

\section{An explicit sampler: description length and the learning coefficient}
\label{sec:main-theorem}

Fix an observed string $z=z_1\cdots z_n\in(\bb{B}^b)^n$, with empirical loss $L_n$ and evidence
$Z_n$ as in Definition~\ref{def:block_empirical_loss}. The string is fixed only for the
probability analysis: the sampler constructed below is built from the model description alone
and inspects neither $z$ nor $n$.

\subsection{Computability assumptions}
\label{sec:main-assumptions}
The following definition collects the effective assumptions under which this section
constructs the sampler $\mathsf{BayesSampler}$ and its
specialisations $B_{\mathfrak m}$. Fix a binary encoding of the finite objects appearing
below: natural numbers, indices $(j,\alpha)$, rational numbers, strings, and tuples of
these. Programs and subroutines receive and return such encodings, and we do not
distinguish a finite object from its encoding.

\begin{definition}
\label{def:computable-bayesian-model}
A \emph{computable Bayesian model} is given by the data:
\begin{itemize}
\item a compact set of parameters $W\subseteq\bb{R}^d$, for some $d\in\bb{Z}_{>0}$,
      contained in a rational box $R$,
\item a statistical model $\{p(\cdot\mid w) \mid w\in R\}$ on $\bb{B}^b$, the set of
      bit-strings of length $b$, for some integer $b>0$,
\item a prior probability measure $\phi$ on $W$ with density $\varphi$,
\end{itemize}
required to satisfy the conditions:
\begin{enumerate}
\item there exists a computable $L_0>0$ such that
      for all $w,w'\in R$ and all $s\in\bb{B}^b$,
      \[
      \bigl|p(s\mid w)-p(s\mid w')\bigr| \le L_0\|w-w'\|;
      \]

\item there exists a computable $\eta\in(0,1]$ such that
      \[
      \eta \le p(s\mid w)\qquad \forall\,w\in R,\ \forall\,s\in\bb{B}^b;
      \]

\item the map $(r,\alpha)\mapsto \phi(D_{r,\alpha})$ is computable;

\item the map $(w,s)\mapsto p(s\mid w)$ is uniformly computable at rational points: there
      is a program which, on input a rational $w\in R$, a string $s\in\bb{B}^b$ and
      $l\in\bb N$, returns a rational $\widehat p$ with $|\widehat p-p(s\mid w)|\le 2^{-l}$.
\end{enumerate}
\end{definition}

\begin{remark}
For \textup{(1)} it suffices that each map $w\mapsto p(s\mid w)$ is $C^1$ on $R$ (uniformly in
$s\in\bb{B}^b$) with a computable uniform bound on $\|\nabla_w p(s\mid w)\|$.
\end{remark}

\begin{definition}
\label{def:model-description}
A \emph{model description} of a computable Bayesian model is a tuple
\[
\mathfrak m:=\bigl(b,\;T_b,\;\phi(D_{\cdot,\cdot}),\;\widehat G,\;R\bigr),
\]
where:
\begin{itemize}
\item $b$ is the string length;
\item $T_b$ is a program witnessing condition \textup{(4)}: on input $(w,s,l)$, with
      $w\in R$ rational, it returns a rational $\widehat p$ with
      $|\widehat p-p(s\mid w)|\le 2^{-l}$;
\item $\phi(D_{\cdot,\cdot})$ is a program computing the function
      $(r,\alpha)\mapsto\phi(D_{r,\alpha})$ of condition \textup{(3)};
\item $\widehat G$ is a rational number with $\widehat G\ge(L_0/\eta)\sqrt d$, where
      $L_0$ and $\eta$ witness conditions \textup{(1)} and \textup{(2)};
\item $R$ is the list of rational endpoints of the box of
      Definition~\ref{def:computable-bayesian-model}.
\end{itemize}
\end{definition}

\begin{definition}
\label{def:cell-tree}
For $j\in\bb N$ and $\alpha\in\{0,\dots,2^j-1\}^d$ set
\[
C_{j,\alpha}:=\prod_{i=1}^d[\alpha_i 2^{-j},(\alpha_i+1)2^{-j}],
\qquad
D_{j,\alpha}:=W\cap C_{j,\alpha}.
\]
A \emph{cell} is a set $D_{j,\alpha}$; it is \emph{positive} if
$\phi(D_{j,\alpha})>0$.
\end{definition}

\begin{remark}
\label{rem:normalisation}
Write $R=\prod_{i=1}^d[a_i,a_i+s_i]$ for the rational box of
Definition~\ref{def:computable-bayesian-model} and let
$A:R\to[0,1]^d$ be the affine map $A(w)_i=(w_i-a_i)/s_i$, which is computable from the
assumptions on $R$. Transporting the model along $A$ gives the model
\[
p'(s\mid w'):=p(s\mid A^{-1}w'),\qquad A^{-1}(w')_i=a_i+s_i w'_i,
\]
on the parameter set $A(W)\subseteq[0,1]^d$ with prior the pushforward $A_*\phi$. Condition
\textup{(1)} holds for $p'$ with the computable constant $L_0\cdot\max_i s_i$, since
\begin{align*}
\bigl|p'(s\mid w'_1)-p'(s\mid w'_2)\bigr|
&\le L_0\|A^{-1}w'_1-A^{-1}w'_2\|\\
&\le L_0\cdot\max_i s_i\cdot\|w'_1-w'_2\|,
\end{align*}
and $b$ and conditions \textup{(2)} and \textup{(4)} are unchanged. The rational
constant $\widehat G$ of Definition~\ref{def:model-description} is replaced by
$\widehat G\cdot\max_i s_i$. The cells of the
transported model are those of Definition~\ref{def:cell-tree} for $A(W)$, namely
$D'_{j,\alpha}=A(W)\cap C_{j,\alpha}$, and the map of condition \textup{(3)} is
\[
(j,\alpha)\longmapsto A_*\phi(D'_{j,\alpha})=\phi\bigl(W\cap A^{-1}(C_{j,\alpha})\bigr).
\]
\end{remark}

In light of Remark~\ref{rem:normalisation}, we assume from now on that $R=[0,1]^d$, so
that $C_{0,\mathbf 0}=R$, $W\subseteq[0,1]^d$, and $D_{0,\mathbf 0}=W$.

\begin{lemma}\label{lem:block_empirical_loss_lipschitz}
Assume conditions \textup{(1)} and \textup{(2)} of
Definition~\ref{def:computable-bayesian-model}. Then for
each $s\in\bb{B}^b$ the map $w\mapsto-\log p(s\mid w)$ is $G$-Lipschitz on $R$, where
\[
G := \frac{L_0}{\eta},
\]
and consequently, for every string $z\in(\bb{B}^b)^n$, the empirical loss $L_n$ of
Definition~\ref{def:block_empirical_loss} is $G$-Lipschitz on $R$.
\end{lemma}

\begin{proof}
On $[\eta,1]$ the map $u\mapsto -\log u$ is $(1/\eta)$-Lipschitz. Thus for any $s\in\bb{B}^b$
and any $w,w'\in R$,
\[
\bigl|\log p(s\mid w)-\log p(s\mid w')\bigr|
\le \frac{1}{\eta}\,\bigl|p(s\mid w)-p(s\mid w')\bigr|
\le \frac{L_0}{\eta}\|w-w'\|.
\]
Averaging over the $n$ strings gives the same constant for $L_n$.
\end{proof}

\begin{definition}
\label{def:schedule}
Define
\[
j(t):=\min\{\,j\in\bb N \mid \widehat G\,2^{-j}\le t^{-2}\,\}
\qquad\text{for }t\ge 1.
\]
\end{definition}

Definition~\ref{def:bayes-sampler} specifies a computable procedure
$\mathsf{BayesSampler}(\mathfrak m,\omega)$, which receives the model description
$\mathfrak m$ deterministically together with an infinite string of bits
$\omega\in\bb{B}^{\bb{N}}$.
It calls two subroutines, which we assume exist with the properties of the following
definition.

\begin{definition}
\label{def:subroutines}
We assume given two subroutines, each consuming input bits sequentially, the decision to
halt depending only on the bits consumed so far:
\begin{itemize}
\item $\mathrm{SampleChild}(D)$: given a positive cell $D$ of a dyadic partition of $W$,
      replaces $D$ by one of its children $D'$, chosen with conditional prior mass
      $\phi(D')/\phi(D)$; it terminates almost surely, and the returned child is almost
      surely positive.
\item $\mathrm{StringSample}(w)$: given a dyadic point $w\in R$, consumes bits to sample one
      string from $p(\cdot\mid w)$; it terminates almost surely.
\end{itemize}
\end{definition}

The random choices in $\mathrm{SampleChild}$ and
$\mathrm{StringSample}$ can be implemented by the interval algorithm of Han and Hoshi
\cite{HanHoshi1997} for sampling a finite distribution from uniformly random bits, using the
programs
$\phi(D_{\cdot,\cdot})$ and $T_b$ of the model description.

\begin{definition}
\label{def:bayes-sampler}
$\mathsf{BayesSampler}$ is the procedure which, on input a model description
$\mathfrak m=(b,T_b,\phi(D_{\cdot,\cdot}),\widehat G,R)$ together with an
infinite string of bits $\omega\in\bb{B}^{\mathbb{N}}$, writes an infinite string (concatenated
$b$-bit strings) on the output tape as follows:
\begin{algorithmic}[1]
\State $D\gets D_{0,\mathbf 0}=W$;\ $j\gets 0$.
\Loop\ \textbf{over} $t=1,2,3,\dots$
  \While{$j<j(t)$}
    \State $D\gets \mathrm{SampleChild}(D)$;\quad $j\gets j+1$.
  \EndWhile
  \State $w_t\gets$ the centre of $C_{j,\alpha}$, where $D=D_{j,\alpha}$.
  \State $\mathrm{StringSample}(w_t)$.
\EndLoop
\end{algorithmic}
\end{definition}

\begin{example}
\label{ex:sampler-run}
Take $b=1$, $d=2$, $W=R=[0,1]^2$, $\phi$ the uniform prior, and
\[
p(1\mid w):=\frac{4+w_1+w_2}{16},
\qquad
p(0\mid w):=1-p(1\mid w).
\]
Then $p(1\mid w)\in[\tfrac14,\tfrac38]$ and $p(0\mid w)\in[\tfrac58,\tfrac34]$, so
$\eta=\tfrac14$ witnesses condition \textup{(2)}, and
\[
|p(s\mid w)-p(s\mid w')|
=\tfrac1{16}\bigl|(w_1-w_1')+(w_2-w_2')\bigr|
\le\tfrac{\sqrt2}{16}\|w-w'\|,
\]
so $L_0=\tfrac{\sqrt2}{16}$ witnesses condition \textup{(1)}. The cell masses are
$\phi(D_{j,\alpha})=4^{-j}$, so conditions \textup{(3)} and \textup{(4)} hold. Then
$G=L_0/\eta=\tfrac{\sqrt2}{4}$, so $G\sqrt d=\tfrac12$; we take $\widehat G=\tfrac12$, and
\[
j(t)=\min\{\,j\in\bb N \mid 2^{j+1}\ge t^2\,\},
\qquad
j(1)=0,\quad j(2)=1,\quad j(3)=3,\quad j(4)=3.
\]
In a run of $\mathsf{BayesSampler}(\mathfrak m,\omega)$: at $t=1$ no new child cell is
sampled, since $j=0=j(1)$, and $z_1$ is sampled from $p(\cdot\mid w_1)$ with
$w_1=(\tfrac12,\tfrac12)$, the centre of $C_{0,\mathbf 0}$; at $t=2$ the loop samples one new child cell, a quadrant
$D_{1,\alpha_1}$ chosen with probability $\tfrac14$, and samples $z_2$ from $p(\cdot\mid w_2)$;
at $t=3$ it samples two new child cells in succession, a cell $D_{2,\alpha_2}$ and
then a cell $D_{3,\alpha_3}$, each chosen with probability $\tfrac14$, and samples
$z_3$ from $p(\cdot\mid w_3)$. Since $j(4)=3$, at time step $t=4$ it
samples $z_4$ from the same $p(\cdot\mid w_3)$. Figure~\ref{fig:sampler-run} shows this run.
\end{example}

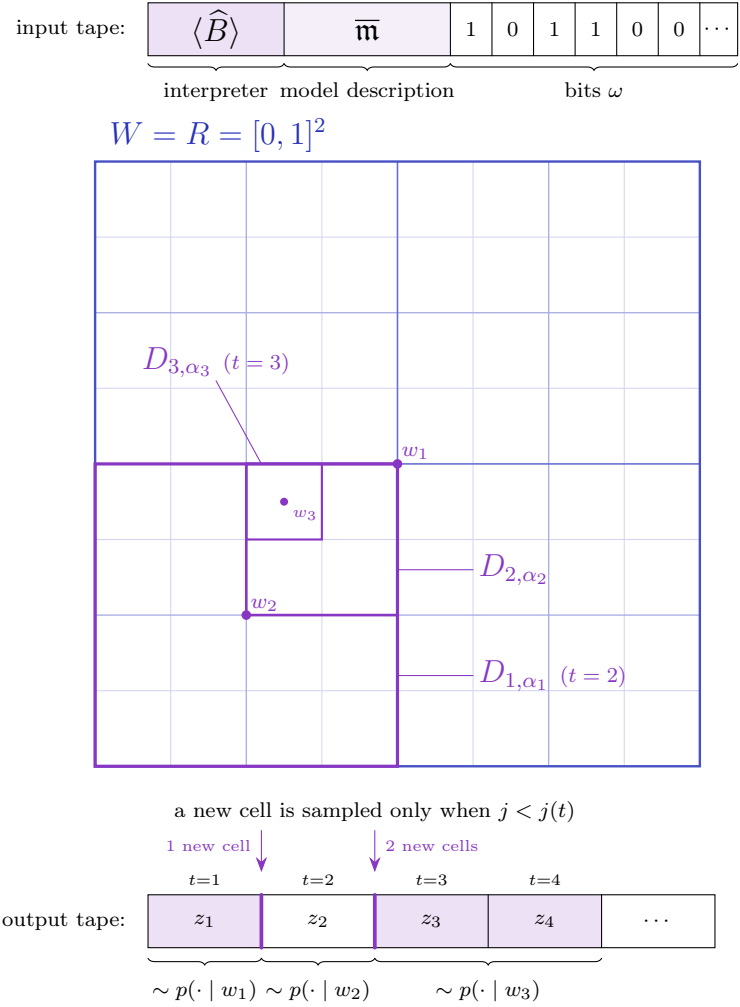
\begin{figure}[tp]
\centering
\begin{tikzpicture}[>={Stealth[length=2mm]}, line cap=round]
\node[anchor=east] at (0.55,9.75) {\scriptsize input tape:};
\fill[cellviolet!15] (0.7,9.4) rectangle (2.5,10.1);
\fill[cellviolet!8] (2.5,9.4) rectangle (4.7,10.1);
\draw (0.7,9.4) rectangle (8.5,10.1);
\draw (2.5,9.4) -- (2.5,10.1);
\draw (4.7,9.4) -- (4.7,10.1);
\node at (1.6,9.75) {$\langle\widehat B\rangle$};
\node at (3.6,9.75) {$\overline{\mathfrak m}$};
\foreach \x in {5.25,5.8,6.35,6.9,7.45,8.0}{ \draw (\x,9.4) -- (\x,10.1); }
\node at (4.975,9.75) {\scriptsize 1};
\node at (5.525,9.75) {\scriptsize 0};
\node at (6.075,9.75) {\scriptsize 1};
\node at (6.625,9.75) {\scriptsize 1};
\node at (7.175,9.75) {\scriptsize 0};
\node at (7.725,9.75) {\scriptsize 0};
\node at (8.25,9.75) {\scriptsize $\cdots$};
\draw[decorate,decoration={brace,mirror,raise=3pt}] (0.7,9.4) -- (2.5,9.4)
  node[midway,below=6pt] {\scriptsize interpreter};
\draw[decorate,decoration={brace,mirror,raise=3pt}] (2.5,9.4) -- (4.7,9.4)
  node[midway,below=6pt] {\scriptsize model description};
\draw[decorate,decoration={brace,mirror,raise=3pt}] (4.7,9.4) -- (8.5,9.4)
  node[midway,below=6pt] {\scriptsize bits $\omega$};
\foreach \x in {1,3,5,7}{ \draw[dotblue!20, line width=0.4pt] (\x,0) -- (\x,8); \draw[dotblue!20, line width=0.4pt] (0,\x) -- (8,\x); }
\foreach \x in {2,6}{ \draw[dotblue!40, line width=0.5pt] (\x,0) -- (\x,8); \draw[dotblue!40, line width=0.5pt] (0,\x) -- (8,\x); }
\draw[dotblue!70, line width=0.6pt] (4,0) -- (4,8);
\draw[dotblue!70, line width=0.6pt] (0,4) -- (8,4);
\draw[dotblue!85, line width=0.9pt] (0,0) rectangle (8,8);
\node[dotblue, anchor=west] at (0.05,8.35) {$W=R=[0,1]^2$};
\draw[cellviolet, line width=1.2pt] (0,0) rectangle (4,4);
\draw[cellviolet, line width=1.0pt] (2,2) rectangle (4,4);
\draw[cellviolet, line width=0.8pt] (2,3) rectangle (3,4);
\draw[cellviolet, thin] (4,1.2) -- (5.0,1.2) node[right, inner sep=2pt] {$D_{1,\alpha_1}$ {\scriptsize($t=2$)}};
\draw[cellviolet, thin] (4,2.6) -- (5.0,2.6) node[right, inner sep=2pt] {$D_{2,\alpha_2}$};
\draw[cellviolet, thin] (2.2,4) -- (1.6,5.1) node[above, inner sep=1pt] {$D_{3,\alpha_3}$ {\scriptsize($t=3$)}};
\fill[cellviolet] (4,4) circle (1.8pt) node[above right, inner sep=1.5pt] {\scriptsize $w_1$};
\fill[cellviolet] (2,2) circle (1.8pt) node[above right, inner sep=1.5pt] {\scriptsize $w_2$};
\fill[cellviolet] (2.5,3.5) circle (1.5pt);
\node[cellviolet, inner sep=1pt] at (2.78,3.32) {\tiny $w_3$};
\node at (3.7,-0.6) {\scriptsize a new cell is sampled only when $j<j(t)$};
\draw[cellviolet, ->] (2.2,-0.85) -- (2.2,-1.4);
\draw[cellviolet, ->] (3.7,-0.85) -- (3.7,-1.4);
\node[cellviolet, anchor=east, inner sep=1pt] at (2.1,-1.05) {\tiny 1 new cell};
\node[cellviolet, anchor=west, inner sep=1pt] at (3.8,-1.05) {\tiny 2 new cells};
\node at (1.45,-1.5) {\tiny $t{=}1$};
\node at (2.95,-1.5) {\tiny $t{=}2$};
\node at (4.45,-1.5) {\tiny $t{=}3$};
\node at (5.95,-1.5) {\tiny $t{=}4$};
\node[anchor=east] at (0.55,-2.05) {\scriptsize output tape:};
\fill[cellviolet!15] (0.7,-2.4) rectangle (2.2,-1.7);
\fill[cellviolet!15] (3.7,-2.4) rectangle (6.7,-1.7);
\draw (0.7,-2.4) rectangle (8.2,-1.7);
\draw (5.2,-2.4) -- (5.2,-1.7);
\draw (6.7,-2.4) -- (6.7,-1.7);
\draw[cellviolet, line width=1.4pt] (2.2,-2.4) -- (2.2,-1.7);
\draw[cellviolet, line width=1.4pt] (3.7,-2.4) -- (3.7,-1.7);
\node at (1.45,-2.05) {\scriptsize $z_1$};
\node at (2.95,-2.05) {\scriptsize $z_2$};
\node at (4.45,-2.05) {\scriptsize $z_3$};
\node at (5.95,-2.05) {\scriptsize $z_4$};
\node at (7.45,-2.05) {\scriptsize $\cdots$};
\draw[decorate,decoration={brace,mirror,raise=4pt}] (0.7,-2.4) -- (2.2,-2.4)
  node[midway,below=8pt] {\scriptsize $\sim p(\cdot\mid w_1)$};
\draw[decorate,decoration={brace,mirror,raise=4pt}] (2.2,-2.4) -- (3.7,-2.4)
  node[midway,below=8pt] {\scriptsize $\sim p(\cdot\mid w_2)$};
\draw[decorate,decoration={brace,mirror,raise=4pt}] (3.7,-2.4) -- (6.7,-2.4)
  node[midway,below=8pt] {\scriptsize $\sim p(\cdot\mid w_3)$};
\end{tikzpicture}
\caption{The run of Example~\ref{ex:sampler-run}, with the input tape segments
$\langle\widehat B\rangle$ and $\overline{\mathfrak m}$ of Definition~\ref{def:interpreter}.
Each new cell is chosen from its parent with conditional prior mass
$\phi(D_{j,\alpha})/\phi(D_{j-1,\alpha'})$, here $\tfrac14$ at every step.}
\label{fig:sampler-run}
\end{figure}

\begin{definition}
\label{def:specialisation}
Let $\mathfrak m$ be a model description. Write $B_{\mathfrak m}$ for the monotone Turing
machine such that
\[
B_{\mathfrak m}(\omega)=\mathsf{BayesSampler}(\mathfrak m,\omega),
\quad\forall\,\omega\in\bb{B}^{\bb{N}}.
\]
\end{definition}

\begin{definition}
\label{def:interpreter}
Fix a prefix-free encoding $\mathfrak m\mapsto\overline{\mathfrak m}$ of model descriptions
as binary strings, and write $\widehat B$ for the monotone Turing machine such that, for
every model description $\mathfrak m$,
\[
\widehat B(\overline{\mathfrak m}\,\omega)=B_{\mathfrak m}(\omega),
\quad\forall\,\omega\in\bb{B}^{\bb{N}}.
\]
\end{definition}

The single machine $\widehat B$ serves every model description, reading
$\overline{\mathfrak m}$ from its input tape before the bits $\omega$. The input tape shown
in Figure~\ref{fig:sampler-run} is that of the universal machine $\mathcal U$ computing
$\mathcal U(\langle\widehat B\rangle\,\overline{\mathfrak m}\,\omega)
=\widehat B(\overline{\mathfrak m}\,\omega)=B_{\mathfrak m}(\omega)$, where
$\langle\widehat B\rangle$ is the encoding of Definition~\ref{def:universal-mtm}.

\begin{lemma}
\label{lem:path-law}
Started from $D_{0,\mathbf 0}=W$ and iterated, $\mathrm{SampleChild}$ produces a nested
sequence of positive cells $D^{(0)}\supseteq D^{(1)}\supseteq\cdots$ with
$\bb{P}(D^{(j)}=D_{j,\alpha})=\phi(D_{j,\alpha})$. The
intersection $\bigcap_j D^{(j)}$ is a single point $w_\infty$ with $w_\infty\sim\phi$, and for every
$j$ the cell $D^{(j)}$ is $\phi$-almost surely the unique cell $D_{j,\alpha}$ containing $w_\infty$.
\end{lemma}

\begin{proof}
By induction on $j$, using that $\mathrm{SampleChild}$ selects a child $D'$ of $D$ with
probability $\phi(D')/\phi(D)$:
$\mathbb P(D^{(j+1)}=D_{j+1,\beta})=\mathbb P(D^{(j)}=D_{j,\alpha})\,\phi(D_{j+1,\beta})/\phi(D_{j,\alpha})
=\phi(D_{j+1,\beta})$ for the parent $\alpha$ of $\beta$ (telescoping). The cells are nested with
$\mathrm{diam}(D^{(j)})\le\sqrt d\,2^{-j}\to 0$, so $\bigcap_j D^{(j)}=\{w_\infty\}$; the law of
$w_\infty$ is the projective limit of the laws of $D^{(j)}$, namely $\phi$. The boundaries of the cubes $C_{j,\alpha}$ are
Lebesgue-null, hence $\phi$-null since $\phi$ has a density, so $\phi$-a.s.\ $w_\infty$
lies, for each $j$, in the interior of exactly one cube $C_{j,\alpha}$; $D^{(j)}$ is
therefore the unique cell $D_{j,\alpha}$ containing $w_\infty$.
\end{proof}

For $n\ge 1$, each positive cell $D_{j(n),\alpha}$ and each time step $t\le n$, write
\[
w_t(\alpha):=\text{the centre of }C_{j(t),\alpha'},
\qquad\text{for the ancestor }D_{j(t),\alpha'}\text{ of }D_{j(n),\alpha}.
\]

\begin{lemma}
\label{lem:sampler-identity}
Write $P_{B_{\mathfrak m}}$ for the semimeasure induced by $B_{\mathfrak m}$
(Definition~\ref{def:program-semimeasure}). For every $n$ and every
$z=z_1\cdots z_n\in(\bb{B}^b)^n$,
\begin{equation}
\label{eq:sampler-sum}
P_{B_{\mathfrak m}}(z)
=\sum_{\alpha}\prod_{t=1}^{n} p\bigl(z_t\mid w_t(\alpha)\bigr)\,\phi\bigl(D_{j(n),\alpha}\bigr),
\end{equation}
summing over the positive cells $D_{j(n),\alpha}$.
\end{lemma}

\begin{proof}
Let $A=\{\omega\in\bb{B}^{\bb{N}} \mid B_{\mathfrak m}(\omega)=z\ast\}$; by
Definition~\ref{def:program-semimeasure},
\[
P_{B_{\mathfrak m}}(z)=\mu(A).
\] By the assumed properties of the subroutines (Definition~\ref{def:subroutines}),
$B_{\mathfrak m}$ writes an infinite string almost surely, and membership of $A$
is determined at time step $t=n$, at which point cells $D_{j(n),\alpha}$ are being sampled.
For a positive cell $D_{j(n),\alpha}$ and $1\le t\le n$ set
\begin{align*}
E_\alpha&:=\bigl\{\omega\in\bb B^{\bb N}\bigm| D^{(j(n))}(\omega)=D_{j(n),\alpha}\bigr\},\\
A_t&:=\bigl\{\omega\in\bb B^{\bb N}\bigm|
B_{\mathfrak m}(\omega)_{(t-1)b+1}\cdots B_{\mathfrak m}(\omega)_{tb}=z_t\bigr\},
\end{align*}
where $D^{(j)}(\omega)$ is the $j$-th cell produced by the successive calls to
$\mathrm{SampleChild}$ in the run of $B_{\mathfrak m}$ on $\omega$, as in
Lemma~\ref{lem:path-law}. The
$E_\alpha$ partition the sample space up to a $\mu$-null set and
$\mu(E_\alpha)=\phi(D_{j(n),\alpha})$ by Lemma~\ref{lem:path-law}, so
\[
\mu(A)=\sum_\alpha \phi(D_{j(n),\alpha})\,\mu(A\mid E_\alpha).
\]
On $E_\alpha$ the current cell at time step $t\le n$ is the ancestor $D_{j(t),\alpha'}$ of
$D_{j(n),\alpha}$, so the $t$-th call to $\mathrm{StringSample}$ samples the $t$-th string
from $p(\cdot\mid w_t(\alpha))$; by the halting assumption of
Definition~\ref{def:subroutines}, it does so on a tail that is i.i.d.\ uniform conditional
on the bits consumed so far. Hence
\begin{align*}
\mu(A\mid E_\alpha)
&=\prod_{t=1}^{n}\frac{\mu\bigl(A_1\cap\dots\cap A_t\bigm|E_\alpha\bigr)}
                      {\mu\bigl(A_1\cap\dots\cap A_{t-1}\bigm|E_\alpha\bigr)}\\
&=\prod_{t=1}^{n}\mu\bigl(A_t\bigm|E_\alpha\cap A_1\cap\dots\cap A_{t-1}\bigr)\\
&=\prod_{t=1}^{n}p(z_t\mid w_t(\alpha)),
\end{align*}
and \eqref{eq:sampler-sum} follows.
\end{proof}

\begin{theorem}
\label{thm:sampler-correct}
Let $\mathfrak m$ be a model description and $B_{\mathfrak m}$ the monotone Turing machine of
Definition~\ref{def:specialisation}.
For every $n$ and every $z=z_1\cdots z_n\in(\bb{B}^b)^n$,
\[
e^{-\pi^2/6}\,Z_n\ \le\ P_{B_{\mathfrak m}}(z)\ \le\ e^{\pi^2/6}\,Z_n.
\]
\end{theorem}

\begin{proof}
Let $D_{j(n),\alpha}$ be a positive cell and $w\in D_{j(n),\alpha}$. For each $t\le n$, both
$w_t(\alpha)$ and $w$ lie in the cube $C_{j(t),\alpha'}$ of the ancestor
$D_{j(t),\alpha'}$ of $D_{j(n),\alpha}$, so
$\|w_t(\alpha)-w\|\le\sqrt d\,2^{-j(t)}$ and, by
Lemma~\ref{lem:block_empirical_loss_lipschitz} and Definition~\ref{def:schedule},
\[
\bigl|\log p(z_t\mid w_t(\alpha))-\log p(z_t\mid w)\bigr|\le G\sqrt d\,2^{-j(t)}\le \widehat G\,2^{-j(t)}\le t^{-2} .
\]
Summing over $t\le n$ and using $\sum_{t\ge 1}t^{-2}=\pi^2/6$,
\[
e^{-\pi^2/6}\prod_{t=1}^n p(z_t\mid w)
\ \le\ \prod_{t=1}^n p(z_t\mid w_t(\alpha))
\ \le\ e^{\pi^2/6}\prod_{t=1}^n p(z_t\mid w).
\]
Integrating over $w\in D_{j(n),\alpha}$ against $\phi$, summing over the positive cells,
which partition $W$ up to a $\phi$-null set, and applying Lemma~\ref{lem:sampler-identity}
together with $Z_n=\sum_{\alpha}\int_{D_{j(n),\alpha}}\prod_{t=1}^n p(z_t\mid w)\,d\phi(w)$,
gives the two bounds.
\end{proof}

\begin{corollary}
\label{cor:sampler-free-energy}
If in addition the hypotheses of Theorem~\ref{thm:main-slt} hold and $X^n=X_1\cdots X_n$ is
drawn i.i.d.\ from the true distribution, then
\[
-\log P_{B_{\mathfrak m}}(X^n)
=
nL_n(w_0)+\lambda\log n-(m-1)\log\log n+O_{\bb P}(1).
\]
\end{corollary}

\begin{proof}
By Theorem~\ref{thm:sampler-correct},
$\bigl|\log P_{B_{\mathfrak m}}(X^n)+\log Z_n\bigr|\le\pi^2/6$; apply
Theorem~\ref{thm:main-slt}.
\end{proof}

\begin{theorem}
\label{thm:universal-aware}
Let $(W,\{p(\cdot\mid w)\}_{w\in W},\phi)$ be a computable Bayesian model
(Definition~\ref{def:computable-bayesian-model}) with model description $\mathfrak m$, and let
$B_{\mathfrak m}$ be the monotone Turing machine of Definition~\ref{def:specialisation}.
For every observed string $z\in(\bb{B}^b)^n$,
\[
M(z)\ \ge\ 2^{-\ell(\langle B_{\mathfrak m}\rangle)}\,e^{-\pi^2/6}\,Z_n,
\]
and consequently
\[
-\log M(z)\ \le\ -\log Z_n + O(1),
\]
the $O(1)$ being independent of $n$ and of the data. If in addition the hypotheses of
Theorem~\ref{thm:main-slt} hold and $X^n=X_1\cdots X_n$ is drawn i.i.d.\ from the true
distribution, then
\[
-\log M(X^n)
\le
nL_n(w_0)
+
\lambda \log n
-
(m-1)\log\log n
+
O_{\bb P}(1).
\]
\end{theorem}

\begin{proof}
By Theorem~\ref{thm:sampler-correct}, $P_{B_{\mathfrak m}}(z)\ge e^{-\pi^2/6}Z_n$.
The machine $B_{\mathfrak m}$ of Definition~\ref{def:specialisation} is
fixed independently of $n$ and of $z$, so dominance (Theorem~\ref{thm:dominance} with
$v=\langle B_{\mathfrak m}\rangle$) gives
$M(z)\ge 2^{-\ell(\langle B_{\mathfrak m}\rangle)}P_{B_{\mathfrak m}}(z)\ge
2^{-\ell(\langle B_{\mathfrak m}\rangle)}e^{-\pi^2/6}Z_n$. Taking $-\log$ gives the second
inequality with $O(1)=\ell(\langle B_{\mathfrak m}\rangle)\log 2+\pi^2/6$, which depends on
$\mathfrak m$ (in particular on $b$) but not on $n$ or the data. Applying
Theorem~\ref{thm:main-slt} to $-\log Z_n$ yields the third.
\end{proof}

\begin{remark}
\label{rem:regular-case}
For regular models the analogue of Theorem~\ref{thm:universal-aware} with the evidence in
place of $M$ is considered in \cite[Thm.~2.3]{ClarkeBarron1990} and
\cite[Thm.~3.74]{Hutter2005}: under suitable hypotheses including the condition that the
Fisher information matrix $I(w_0)$ is nonsingular at a point $w_0\in W$, for
$X_1,\dots,X_n$ drawn i.i.d.\ from $p(\cdot\mid w_0)$,
\[
\bb E\left[\log\frac{\prod_{i=1}^n p(X_i\mid w_0)}{Z_n}\right]
\ \le\
\log\frac{1}{\varphi(w_0)}+\frac{d}{2}\log\frac{n}{2\pi}+\frac{1}{2}\log\det I(w_0)+o(1).
\]
The bound concerns the evidence $Z_n$ and not $M$; to bound $-\log M$ one must exhibit a
machine whose semimeasure is comparable to $Z_n$. For continuously parameterized classes
this step is discussed in \cite[Sec.~3.7.2]{Hutter2005}; Theorem~\ref{thm:universal-aware}
proves it, for singular models.
\end{remark}

\begin{remark}
\label{rem:llc-mdl}
In \cite{Urdshals2025smdl} a two-part code is constructed: a code word
$\llbracket p\rrbracket$ describing a distribution $p$ in the model class is transmitted
first, followed by the message $x^n=x_1\cdots x_n$ encoded with respect to $p$, written
$\llbracket x^n\rrbracket_p$, the combined length being
$\ell(\llbracket p\rrbracket)+\sum_{i=1}^n\log\frac{1}{p(x_i)}$. The \emph{redundancy} of
the code is its excess length over the encoding that uses the data distribution itself: for
$x_1,\dots,x_n$ drawn i.i.d.\ from $q$,
\[
R_n\;:=\;\ell(\llbracket p\rrbracket)+\ell(\llbracket x^n\rrbracket_p)
-\sum_{i=1}^n\log\frac{1}{q(x_i)}
\;=\;\ell(\llbracket p\rrbracket)+\sum_{i=1}^n\log\frac{q(x_i)}{p(x_i)}\,.
\]
The main theorem of \cite{Urdshals2025smdl} asserts that the code can be chosen so that,
for every realizable data generating distribution $q$ in the model class,
\[
R_n\;=\;\lambda\log n-(m-1)\log\log n+O_{\bb P}(1),
\]
with $\lambda$ the learning coefficient of $q$ for the model and $m$ its multiplicity. This
two-part code is the core inspiration for the construction of $B_{\mathfrak m}$: our
machine is a generic algorithm, one for every computable Bayesian model,
which, when its input tape is filled with uniformly random bits, produces the output of a
code of this format.
\end{remark}

\begin{remark}
\label{rem:related-work}
We can represent the sequence of positive cells sampled by $B_{\mathfrak m}$ as a directed
graphical model: the nodes are the positive cells, and each edge is labelled by the
probability of its target being sampled. Doing this exhibits a similarity, shown in
Figure~\ref{fig:two-trees}, between our
construction and that of \cite[Lem.~5.3]{HennickDellago2026}, whose trees have their
branches pruned according to the following rule: for a threshold $t>0$, the node
$D_{j,\alpha}$ is removed, together with its descendants, if and only if
\[
\{\,w\in W \mid K(w)<t\,\}\cap D_{j,\alpha}=\varnothing,
\]
with $K$ as in Definition~\ref{def:HK-binary}; the loss of \cite{HennickDellago2026} is
the KL-divergence $\widehat K$, which agrees with $K$ in the realizable case
(Remark~\ref{rem:rfv-uniqueness}). Writing $\mathcal T_t$ for the tree of
retained cells and $\partial\mathcal T_t$ for its set of infinite branches, the map sending
a parameter to its branch of nested cells pushes the prior forward to a measure $\phi^*$
with $\phi^*(\partial\mathcal T_t)=\phi(K<t)$, and the learning coefficient $\lambda$ of
the model, Definition~\ref{def:learning-coefficient}, is there shown to be the scaling
exponent of this measured boundary as $t\to 0$ \cite[Lem.~5.6]{HennickDellago2026}:
\[
\lambda\ =\ \lim_{t\to 0}\ \frac{\log \phi^*(\partial\mathcal T_t)}{\log t}.
\]
\end{remark}

\begin{figure}[tp]
\centering
\resizebox{\textwidth}{!}{%
\begin{tikzpicture}[>={Stealth[length=1.8mm]}, line cap=round]
\node at (3.6,5.3) {\small The tree of positive cells.};
\node (r) at (3.6,4.6) {$D_{0,0}$};
\node (a) at (1.2,2.3) {$D_{1,0}$};
\node (b) at (6.0,2.3) {$D_{1,1}$};
\node (c) at (0.0,0) {$D_{2,0}$};
\node (d) at (2.4,0) {$D_{2,1}$};
\node (e) at (4.8,0) {$D_{2,2}$};
\node (f) at (7.2,0) {$D_{2,3}$};
\draw[->] (r) -- (a) node[midway, left=1pt] {$\scriptstyle\frac{\phi(D_{1,0})}{\phi(D_{0,0})}$};
\draw[->] (r) -- (b) node[midway, right=1pt] {$\scriptstyle\frac{\phi(D_{1,1})}{\phi(D_{0,0})}$};
\draw[->] (a) -- (c) node[midway, left=1pt] {$\scriptstyle\frac{\phi(D_{2,0})}{\phi(D_{1,0})}$};
\draw[->] (a) -- (d) node[midway, right=1pt] {$\scriptstyle\frac{\phi(D_{2,1})}{\phi(D_{1,0})}$};
\draw[->] (b) -- (e) node[midway, left=1pt] {$\scriptstyle\frac{\phi(D_{2,2})}{\phi(D_{1,1})}$};
\draw[->] (b) -- (f) node[midway, right=1pt] {$\scriptstyle\frac{\phi(D_{2,3})}{\phi(D_{1,1})}$};
\foreach \x in {0.0,2.4,4.8,7.2}{ \node at (\x,-0.6) {$\vdots$}; }
\node at (12.8,5.3) {\small The tree of Hennick--Dellago at threshold $t$.};
\node (rr) at (12.8,4.6) {$D_{0,0}$};
\node (aa) at (10.4,2.3) {$D_{1,0}$};
\node (bb) at (15.2,2.3) {$D_{1,1}$};
\node (cc) at (9.2,0) {$D_{2,0}$};
\node[gray!70] (dd) at (11.6,0) {$D_{2,1}$};
\node[gray!70] (ee) at (14.0,0) {$D_{2,2}$};
\node (ff) at (16.4,0) {$D_{2,3}$};
\draw[->] (rr) -- (aa);
\draw[->] (rr) -- (bb);
\draw[->] (aa) -- (cc);
\draw[->, dashed, gray!70] (aa) -- (dd);
\draw[->, dashed, gray!70] (bb) -- (ee);
\draw[->] (bb) -- (ff);
\node at (9.2,-0.6) {$\vdots$};
\node at (16.4,-0.6) {$\vdots$};
\end{tikzpicture}%
}
\caption{Left: a directed graphical model presentation of the sequences of nested positive
cells that may be taken, along with their probabilities, by the algorithm $B_{\mathfrak m}$
of Definition~\ref{def:specialisation}, for dimension $d=1$ and compact parameter set
$W=[0,1]$. Right: the tree construction of \cite[Lem.~5.3]{HennickDellago2026}, which for
$t>0$ removes the nodes $D_{j,\alpha}$ (grey) such that
$\{w\in W \mid K(w)<t\}\cap D_{j,\alpha}=\varnothing$, along with their descendants.}
\label{fig:two-trees}
\end{figure}
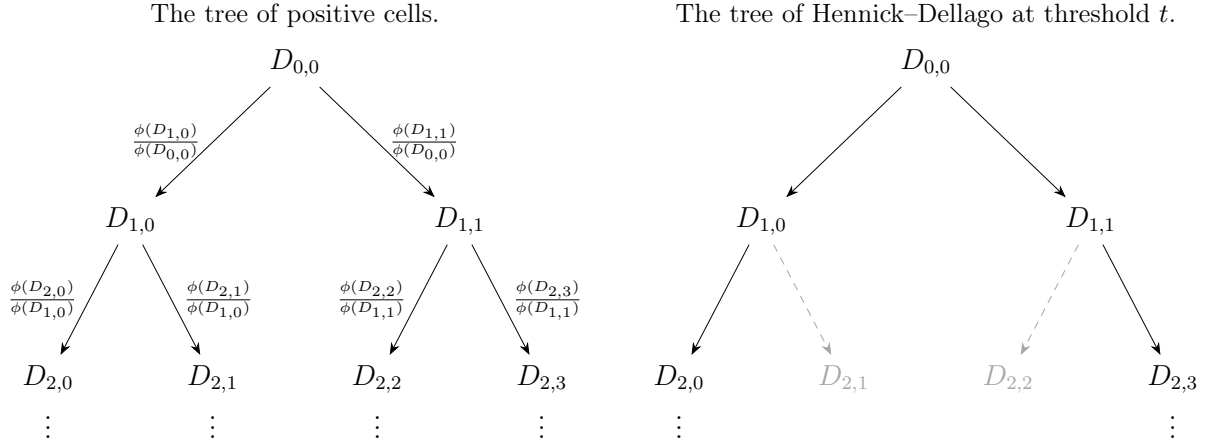

\section*{Acknowledgements}

The majority of this work was carried out while the authors were at Timaeus.
This project is funded by the Advanced Research + Invention Agency (ARIA).

\appendix

\section{The Solomonoff distribution}
\label{sec:setup:universal}

Following \cite{Hutter:24uaibook2} we recall the elements of algorithmic information theory.

\begin{definition}
\label{def:binary-strings}
Let $\bb{B}:=\{0,1\}$.  Write $\bb{B}^\ast$ for finite binary strings and $\bb{B}^{\bb{N}}$ for
infinite binary strings.
For $x,y\in\bb{B}^\ast\cup\bb{B}^{\bb{N}}$, write $x\sqsubseteq y$ if $x$ is a \emph{prefix} of $y$.
\end{definition}

\begin{definition}
\label{def:cylinders-measure}
For $x\in\bb{B}^\ast$, define the \emph{cylinder set}
\[
\Gamma_x:=\{\omega\in\bb{B}^{\bb{N}} \mid x\sqsubseteq \omega\}.
\]
Let $\mu$ denote the \emph{uniform product measure} on $\bb{B}^{\bb{N}}$, characterised by
$\mu(\Gamma_x)=2^{-\ell(x)}$, where $\ell(x)$ denotes the length of $x$.
\end{definition}

\begin{definition}[{cf.\ \cite[Definition~2.7.1]{Hutter:24uaibook2}}]
\label{def:monotone-tm}
A \emph{monotone Turing machine} is given by the data:
\begin{itemize}
\item a state set $\{1,\dots,k\}$ for some $k\ge 1$,
\item one input tape, one output tape, and $j\ge 1$ work tapes,
\item a transition function $\delta$,
\end{itemize}
required to satisfy the conditions:
\begin{itemize}
\item all tapes are binary, with no blank symbol, and the work tapes are initialised with zeros,
\item the input tape is read-only and unidirectional (its head can only move from left to right),
\item the output tape is write-only and unidirectional.
\end{itemize}
For $\omega\in\bb{B}^{\bb{N}}$ we write $T(\omega)\in\bb{B}^\ast\cup\bb{B}^{\bb{N}}$ for the
output written during the run of $T$ on input tape $\omega$. For $v,y\in\bb{B}^\ast$ we write
$T(v)=y\ast$ if, in the run on any input
$\omega\sqsupseteq v$, the input head has consumed exactly $v$ at the step in which the
last bit of $y$ is written.
\end{definition}

\begin{definition}
\label{def:program-semimeasure}
Let $T$ be a monotone Turing machine. For $y\in\bb{B}^\ast$, the set of \emph{minimal
programs} \cite[Definition~2.7.1]{Hutter:24uaibook2}
\[
\mathcal V_{T,y}:=\{p\in\bb{B}^\ast \mid T(p)=y\ast\}
\]
is prefix-free, and
\[
\{\omega\in\bb{B}^{\bb{N}} \mid y\sqsubseteq T(\omega)\}
=
\bigsqcup_{p\in \mathcal V_{T,y}}\Gamma_p
\]
\cite[Section~3.8]{Hutter:24uaibook2}. Define the \emph{semimeasure} induced by $T$ on
cylinders by
\begin{equation}
\label{eq:probability}
P_T(\Gamma_y)
:=
\mu\bigl(\{\omega\in\bb{B}^{\bb{N}} \mid y\sqsubseteq T(\omega)\}\bigr)
=
\sum_{p\in \mathcal V_{T,y}} 2^{-\ell(p)},
\end{equation}
the probability that the output of $T$ begins with $y$ when the input tape is filled with
uniformly random bits. We abbreviate $P_T(y):=P_T(\Gamma_y)$.
\end{definition}

\begin{definition}[{cf.\ \cite[Definition~2.7.2]{Hutter:24uaibook2}}]
\label{def:universal-mtm}
Every monotone Turing machine has a canonical prefix-free encoding $T\mapsto\langle T\rangle$
as a binary string. A monotone Turing
machine $\mathcal U$ is \emph{universal} if
\[
\mathcal{U}(\langle T\rangle\omega)=T(\omega)
\qquad\text{for every monotone Turing machine $T$ and every $\omega\in\bb{B}^{\bb{N}}$.}
\]
Universal machines exist \cite[Definition~2.7.2]{Hutter:24uaibook2}; we fix one, the
\emph{reference machine} $\mathcal U$ \cite[Remark~2.7.7]{Hutter:24uaibook2}.
\end{definition}

\begin{definition}[{cf.\ \cite[Definition~3.8.2]{Hutter:24uaibook2}}]
\label{def:universal-semimeasure}
The \emph{Solomonoff distribution} (or universal a priori probability) is
\[
M(\Gamma_y) := \sum_{p\,:\,\mathcal U(p)=y\ast} 2^{-\ell(p)},
\]
summing over the minimal programs $\mathcal V_{\mathcal U,y}$ of
Definition~\ref{def:program-semimeasure}. By \eqref{eq:probability},
$M(\Gamma_y)=P_{\mathcal U}(\Gamma_y)$: the probability that the output of $\mathcal U$
begins with $y$ on a uniformly random input tape. For $y\in\bb{B}^\ast$ we abbreviate
$M(y):=M(\Gamma_y)$, as in \eqref{eq:M_sum}.
\end{definition}

\begin{theorem}
\label{thm:dominance}
The universal semimeasure $M$ dominates the contribution of every input prefix
\cite{Hutter2005,Hutter:24uaibook2}: for all $v,y\in\bb{B}^\ast$,
\[
M(y) \;\ge\; 2^{-\ell(v)}\,
\mu\bigl(\{\omega\in\bb{B}^{\bb{N}} \mid y\sqsubseteq \mathcal U(v\omega)\}\bigr).
\]
In particular, taking $v=\langle T\rangle$ for a monotone Turing machine $T$ gives
$M(y)\ge 2^{-\ell(\langle T\rangle)}\,P_T(y)$.
\end{theorem}

\section{Singular learning theory (SLT)}
\label{sec:setup:slt}

The asymptotics of the free energy associated to the Bayesian evidence are (for singular
models) given by the theorems of Watanabe
\cite{watanabeAlgebraicGeometryStatistical2009,Watanabe2018}. Since then a range of
introductions to and developments of the field have appeared
\cite{wei2022deep,lin2011algebraic,waring2021geometric,lau2024locallearningcoefficientsingularityaware,hoogland2024developmental,wang2024differentiation,chen2023tms1,murfet2025programssingularities}.
We give a minimal account here.

Throughout, $\log$ denotes the natural logarithm and $\log_2$ denotes base-$2$ logarithm.
\medskip

We work in the binary setting used throughout the paper. Fix:
\begin{itemize}
\item a dimension $d\in\bb{Z}_{>0}$ and a compact parameter set $W\subseteq\bb{R}^d$,
\item a statistical model $p(x\vert w)$ for $x\in\bb{B}$, $w\in W$,
\item a true distribution $q(x)$ for $x\in\bb{B}$,
\item a prior probability measure $\phi$ on $W$ with density $\varphi$ (so $d\phi(w)=\varphi(w)\,dw$),
\item an integer $n\ge 1$,
\item data $X_1,\dots,X_n\overset{\mathrm{i.i.d.}}{\sim} q$.
\end{itemize}

\begin{definition}
\label{def:model-true-ldr}
Define the \emph{log density ratio}:
\[
f(x,w):=\log\frac{q(x)}{p(x\vert w)},
\qquad (x\in\bb{B},\ w\in W).
\]
We implicitly assume the usual absolute-continuity condition
$q(x)>0 \Rightarrow p(x\vert w)>0$ for all $w\in W$ so that $f$ is well-defined.
\end{definition}

\begin{definition}
\label{def:losses-KL}
Define the \emph{population loss} and \emph{empirical loss}:
\[
L(w):=-\bb{E}\big[\log p(X\vert w)\big],
\qquad
L_n(w):=-\frac1n\sum_{i=1}^n \log p(X_i\vert w),
\]
where the expectation is taken under $q$.

Let $w_0\in W$ be a minimizer of $L(w)$ (one exists since $L$ is continuous on the compact set $W$ under fundamental condition (I) below) and define the minimizer set
\[
W_0:=\{w\in W \mid L(w)=L(w_0)\}.
\]
\end{definition}

\begin{definition}
\label{def:HK-binary}
For $w\in W$, define the KL-divergence
\[
\widehat K(w):=\sum_{x\in\bb{B}}q(x)\log\frac{q(x)}{p(x\vert w)}
=\bb{E}\big[f(X,w)\big].
\]
Define the \emph{(centered) KL function}
\[
K(w):=\widehat K(w)-\widehat K(w_0)
=L(w)-L(w_0)\ge 0,
\]
so that $K(w)=0$ if and only if $w\in W_0$.
\end{definition}

\begin{definition}
\label{def:empirical-KL}
Define the \emph{empirical KL-divergence} and its centered version:
\[
\widehat K_n(w)
:=
\frac1n\sum_{i=1}^n f(X_i,w)
=
\frac1n\sum_{i=1}^n \log\frac{q(X_i)}{p(X_i\vert w)},
\]
\[
K_n(w):=\widehat K_n(w)-\widehat K_n(w_0)
=L_n(w)-L_n(w_0),
\]
so that $K_n(w_0)=0$.
\end{definition}

\begin{definition}
\label{def:evidence}
The \emph{partition function} (Bayesian evidence) is
\[
Z_n:=\int_W \prod_{i=1}^n p(X_i\vert w)\,\varphi(w)\,dw
=\int_W \exp\!\big(-nL_n(w)\big)\,\varphi(w)\,dw.
\]
Equivalently,
\[
Z_n
=
\exp\!\big(-nL_n(w_0)\big)
\int_W \exp\!\big(-nK_n(w)\big)\,\varphi(w)\,dw.
\]
\end{definition}

\begin{definition}
\label{def:learning-coefficient}
Define the \emph{zeta function} for $\Re(z)>0$, where $\Re(z)$ denotes the real part of $z$,
by
\[
\zeta(z):=\int_W K(w)^z\,\varphi(w)\,dw;
\]
under fundamental conditions (I) and (II) below it admits a meromorphic continuation to $\bb{C}$.
Let $-\lambda$ be the pole of $\zeta$ with largest real part, and let $m$ be its order.
Then $\lambda>0$ is the \emph{learning coefficient} and $m$ its \emph{multiplicity}.
\end{definition}

The learning-theoretic asymptotics depend on (i) analyticity of the log density ratio and
(ii) mild integrability and regularity assumptions on $(W,\varphi)$.

\begin{definition}
\label{def:fundamental-I}
Fix $s\ge 2$. The tuple $(q,p,\varphi)$ satisfies \emph{fundamental condition (I) with index $s$}
if the following hold.
\begin{enumerate}
\item There exists an open set $W^{(\mathbb C)}\subset\mathbb C^d$ such that
\begin{enumerate}
\item[(a)] $W\subset W^{(\mathbb C)}\cap\mathbb R^d$,
\item[(b)] the map $w\mapsto f(\cdot,w)$ is an $L^s(q)$-valued complex-analytic function on
$W^{(\mathbb C)}$,
\item[(c)] the envelope $\bar f(x):=\sup_{w\in W^{(\mathbb C)}}|f(x,w)|$ satisfies $\bar f\in L^s(q)$.
\end{enumerate}
\item There exists $\varepsilon_1>0$ such that, defining
\[
Q(x):=\sup\{\,p(x\vert w) \mid w\in W,\ K(w)\le \varepsilon_1\,\},
\]
we have the integrability condition
\[
\sum_{x\in\bb{B}} \bar f(x)^2\,Q(x)\,q(x)<\infty.
\]
\end{enumerate}
\end{definition}

\begin{remark}
On the finite space $\bb{B}$, item 2 of fundamental condition (I) follows from item
1(c), since $Q\le 1$ and the sum is finite; we retain it as stated to match the general form
of \cite{Watanabe2018}.
\end{remark}

\begin{definition}
\label{def:fundamental-II}
The tuple $(W,\varphi)$ satisfies \emph{fundamental condition (II)} if:
\begin{enumerate}
\item $W$ is a compact semianalytic set of the form
\[
W=\{w\in\mathbb R^d \mid \pi_1(w)\ge 0,\dots,\pi_J(w)\ge 0\},
\]
where each $\pi_j$ is real-analytic on some open neighborhood $W^{(\mathbb R)}\subset\mathbb R^d$
containing $W$.
\item The prior density factors as $\varphi(w)=\varphi_1(w)\varphi_2(w)$, where
$\varphi_1(w)>0$ is $C^\infty$ on $W^{(\mathbb R)}$ and $\varphi_2(w)\ge 0$ is real-analytic on
$W^{(\mathbb R)}$.
\end{enumerate}
\end{definition}

\begin{definition}
\label{def:rfv}
Define
\[
f_0(x,w):=\log\frac{p(x\vert w_0)}{p(x\vert w)}
=f(x,w)-f(x,w_0),
\qquad (x\in\bb{B},\ w\in W),
\]
so that $\bb{E}\big[f_0(X,w)\big]=K(w)$, where $X\sim q$.
The model has \emph{relatively finite variance} if there exists $c_0>0$ such that
\[
\bb{E}\big[f_0(X,w)^2\big]\le c_0\,K(w)
\qquad\text{for all }w\in W.
\]
\end{definition}

\begin{remark}
\label{rem:rfv-uniqueness}
For $w\in W_0$ the condition forces $\bb{E}\big[f_0(X,w)^2\big]=0$, so all optimal parameters
induce the same distribution $q$-almost surely; in particular Definition~\ref{def:rfv}
does not depend on the choice of $w_0$. In the realizable case $q(x)=p(x\vert w_0)$ we
have $f_0=f$ almost surely and $K=\widehat K$.
\end{remark}

\begin{theorem}
\label{thm:main-slt}
Assume the standard SLT hypotheses:
fundamental condition (I) with some index $s\ge 2$,
fundamental condition (II),
and relatively finite variance (Definition~\ref{def:rfv}); assume moreover that
$\varphi\bigl(\{w\in W \mid 0<K(w)\le\varepsilon\}\bigr)>0$ for every $\varepsilon>0$, so that
the pole defining $(\lambda,m)$ exists.
Let $(\lambda,m)$ be the learning coefficient and multiplicity from
Definition~\ref{def:learning-coefficient}.
Then the normalized free energy satisfies
\[
F_n^{(0)}
:=
-\log\!\int_W \exp\!\big(-nK_n(w)\big)\,\varphi(w)\,dw
=
\lambda \log n -(m-1)\log\log n +O_{\bb{P}}(1).
\]
Consequently,
\[
-\log Z_n
=
nL_n(w_0)+\lambda \log n-(m-1)\log\log n+O_{\bb{P}}(1).
\]
\end{theorem}

Fix an integer $b\ge 1$.
We identify $s\in\bb{B}^b$ with the integer $\sum_{k=1}^{b}s_k\,2^{\,b-k}\in\{0,1,\dots,2^b-1\}$ of which
it is the length-$b$ binary representation; thus $\bb{B}^b$ is a finite outcome alphabet of $2^b$
symbols, and $p(\cdot\mid w)$ is a distribution on it.
In the body of the paper the observed data is treated as a string
$z = z_1\cdots z_n \in (\bb{B}^b)^n$ (i.e.\ $n$ i.i.d.\ samples, each encoded in $b$ bits).
Whenever $M(z)$ appears, $z$ is understood as the \emph{concatenated} binary string
$z_1\cdots z_n \in \bb{B}^{nb}$.
The case $b=1$ recovers the binary setting. The development of this section is stated for
binary outcomes $x\in\bb{B}$ but uses only finiteness of the outcome alphabet, and holds
verbatim over $\bb{B}^b$; we invoke Theorem~\ref{thm:main-slt} in this generality.

\begin{definition}
\label{def:block_empirical_loss}
Fix an observed string $z=z_1\cdots z_n\in (\bb{B}^b)^n$.
Define the empirical loss
\[
L_n(w) := -\frac{1}{n}\log\Bigl(\prod_{i=1}^n p(z_i\mid w)\Bigr),
\qquad
e^{-nL_n(w)}=\prod_{i=1}^n p(z_i\mid w),
\]
and the partition function (Bayesian evidence)
\[
Z_n := \int_W e^{-nL_n(w)}\,d\phi(w).
\]
\end{definition}

For $b=1$, Definition~\ref{def:block_empirical_loss} agrees with Definitions~\ref{def:losses-KL} and~\ref{def:evidence}.

\renewcommand{\refname}{Bibliography}
\bibliographystyle{plain}
\bibliography{bibliography}

@article{ClarkeBarron1990,
  author  = {Clarke, Bertrand S. and Barron, Andrew R.},
  title   = {Information-theoretic asymptotics of {B}ayes methods},
  journal = {IEEE Transactions on Information Theory},
  volume  = {36},
  number  = {3},
  pages   = {453--471},
  year    = {1990}
}

@book{watanabeAlgebraicGeometryStatistical2009,
  title     = {Algebraic {{Geometry}} and {{Statistical Learning Theory}}},
  author    = {Watanabe, Sumio},
  year      = {2009},
  publisher = {{Cambridge University Press}},
  address   = {{USA}}
}

@misc{murfet2025programssingularities,
  title         = {Programs as Singularities},
  author        = {Daniel Murfet and Will Troiani},
  year          = {2025},
  eprint        = {2504.08075},
  archivePrefix = {arXiv},
  primaryClass  = {cs.LO},
  url           = {https://arxiv.org/abs/2504.08075},
  note          = {arXiv:2504.08075}
}

@book{Hutter2005,
  author    = {Marcus Hutter},
  title     = {Universal Artificial Intelligence: Sequential Decisions based on Algorithmic Probability},
  publisher = {Springer},
  address   = {Berlin},
  year      = {2005}
}

@Book{Hutter:24uaibook2,
  author =       "Marcus Hutter and David Quarel and Elliot Catt",
  title =        "An Introduction to Universal Artificial Intelligence",
  series =       "Chapman \& Hall/CRC Artificial Intelligence and Robotics Series",
  publisher =    "Taylor and Francis",
  _month =       may,
  year =         "2024",
  isbn =         "Paperback:9781032607023, Harcover:9781032607153, eBook:9781003460299",
  pages =        "500",
  bibtex =       "http://www.hutter1.net/official/bib.htm#uaibook2",
  doi =          "10.1201/9781003460299",
  _note =         "500+ pages, http://www.hutter1.net/ai/uaibook2.htm",
  url =          "http://www.hutter1.net/ai/uaibook2.htm",
  http =         "http://www.routledge.com/An-Introduction-to-Universal-Artificial-Intelligence/Hutter-Catt-Quarel/p/book/9781032607023",
  pdf =          "http://www.hutter1.net/publ/uaibook2.pdf",
  slides =       "http://www.hutter1.net/ai/suaibook.pdf",
  video =        "http://cartesiancafe.podbean.com/e/marcus-hutter-universal-artificial-intelligence-and-solomonoff-induction/",
  keywords =     "Artificial general intelligence; algorithmic information theory;
                  Bayes mixture distributions; universal sequence prediction;
                  context tree weighting; rational agents; sequential decision theory;
                  universal intelligent agents; reinforcement learning;
                  games and multi-agent systems; approximation/implementation/application;
                  AGI-safety; philosophy of AI.",
  abstract =     "`An Introduction to Universal Artificial Intelligence'
                  provides the formal underpinning of what it means for an agent 
                  to act intelligently in an unknown environment. 
                  First presented in `Universal Algorithmic Intelligence' (Hutter, 2000), 
                  UAI offers a framework in which virtually all AI problems can be formulated, 
                  and a theory of how to solve them. 
                  UAI unifies ideas from sequential decision theory, 
                  Bayesian inference, and algorithmic information theory to construct AIXI, 
                  an optimal reinforcement learning agent 
                  that learns to act optimally in unknown environments. 
                  AIXI is the theoretical gold standard for intelligent behavior.
                      The book covers both the theoretical and practical aspects of UAI. 
                  Bayesian updating can be done efficiently with context tree weighting, 
                  and planning can be approximated by sampling with Monte Carlo tree search. 
                  It provides algorithms for the reader to implement, 
                  and experimental results to compare against. 
                  These algorithms are used to approximate AIXI. 
                  The book ends with a philosophical discussion of Artificial General Intelligence: 
                  Can super-intelligent agents even be constructed? 
                  Is it inevitable that they will be constructed,
                  and what are the potential consequences?
                      This text is suitable for late undergraduate students. 
                  It provides an extensive chapter to fill in the required 
                  mathematics, probability, information, 
                  and computability theory background.",
  support =      "ARC grant DP150104590",
  for =          "010404(20%),080101(20%),080198(20%),080299(10%),080401(30%)",
}

@book{Watanabe2018,
  author    = {Sumio Watanabe},
  title     = {Mathematical Theory of Bayesian Statistics},
  publisher = {CRC Press},
  address   = {Boca Raton, FL},
  year      = {2018}
}

@article{hoogland2024developmental,
  title   = {Loss Landscape Degeneracy and Stagewise Development in Transformers},
  author  = {Jesse Hoogland and George Wang and Matthew Farrugia-Roberts and Liam Carroll and Susan Wei and Daniel Murfet},
  year    = {2025},
  journal = {Transactions on Machine Learning Research},
  url     = {https://arxiv.org/abs/2402.02364}
}

@article{wang2024differentiation,
  title   = {Differentiation and Specialization of Attention Heads via the Refined Local Learning Coefficient},
  author  = {George Wang and Jesse Hoogland and Stan van Wingerden and Zach Furman and Daniel Murfet},
  year    = {2024},
  journal = {arXiv preprint arXiv:2410.02984},
  url     = {https://arxiv.org/abs/2410.02984}
}

@phdthesis{lin2011algebraic,
  title  = {Algebraic methods for evaluating integrals in {B}ayesian statistics},
  author = {Lin, Shaowei},
  year   = {2011},
  school = {University of California, Berkeley}
}

@article{lau2024locallearningcoefficientsingularityaware,
  title   = {The Local Learning Coefficient: A Singularity-Aware Complexity Measure},
  author  = {Lau, Edmund and Furman, Zach and Wang, George and Murfet, Daniel and Wei, Susan},
  journal = {arXiv preprint arXiv:2308.12108},
  year    = {2024},
  url     = {https://arxiv.org/abs/2308.12108}
}

@article{wei2022deep,
  title     = {Deep {L}earning is {S}ingular, and {T}hat’s {G}ood},
  author    = {Wei, Susan and Murfet, Daniel and Gong, Mingming and Li, Hui and Gell-Redman, Jesse and Quella, Thomas},
  journal   = {IEEE Transactions on Neural Networks and Learning Systems},
  year      = {2022},
  publisher = {IEEE}
}

@article{chen2023tms1,
  title   = {Dynamical versus {B}ayesian Phase Transitions in a Toy Model of Superposition},
  author  = {Zhongtian Chen and Edmund Lau and Jake Mendel and Susan Wei and Daniel Murfet},
  year    = {2023},
  journal = {arXiv preprint arXiv:2310.06301},
  url     = {https://arxiv.org/abs/2310.06301}
}

@mastersthesis{waring2021geometric,
  title  = {Geometric perspectives on program synthesis and semantics},
  author = {Waring, Thomas},
  year   = {2021},
  school = {The University of Melbourne}
}

@misc{Urdshals2025smdl,
  author        = {Urdshals, Einar and Lau, Edmund and Hoogland, Jesse and van Wingerden, Stan and Murfet, Daniel},
  title         = {Compressibility Measures Complexity: Minimum Description Length Meets Singular Learning Theory},
  year          = {2025},
  eprint        = {2510.12077},
  archivePrefix = {arXiv},
  primaryClass  = {cs.LG},
  note          = {arXiv:2510.12077}
}

@article{HanHoshi1997,
  author  = {Han, Te Sun and Hoshi, Mamoru},
  title   = {Interval algorithm for random number generation},
  journal = {IEEE Transactions on Information Theory},
  volume  = {43},
  number  = {2},
  pages   = {599--611},
  year    = {1997}
}

@misc{HennickDellago2026,
  author = {Hennick, Max and Dellago, Matthias},
  title  = {Learning coefficients, fractals, and trees in parameter space},
  year   = {2026},
  note   = {Preprint, March 2026, \url{https://openreview.net/forum?id=KUFH0n1BIM}}
}

\end{document}